\pdfoutput=1

\documentclass[a4paper,11pt]{article}
\usepackage{slashed}
\usepackage[affil-sl,auth-sc]{authblk}
\usepackage{amssymb,amsmath,amsbsy}
\usepackage{mathrsfs}
\usepackage{mathbbol}
\usepackage{bm}
\usepackage{appendix}
\usepackage{hyperref}
\usepackage[top=1.2in, bottom=1.2in,left=0.9in,includefoot]{geometry}               
\usepackage{graphicx}
\usepackage{cite}
\usepackage{float}
\usepackage{caption}
\usepackage{wasysym}
\usepackage{amsthm}

\def\code{{\tt SUSY\_FLAVOR}}

\newif\ifContLineOne
\newif\ifContLineTwo
\newif\ifContLineThree

\def\conC#1{\vbox{\ialign{##\crcr
  \ifContLineThree\hrulefill\else\vphantom{\hrulefill}\fi\crcr
  \noalign{\kern3.2pt\nointerlineskip}
  \ifContLineTwo\hrulefill\else\vphantom{\hrulefill}\fi\crcr
  \noalign{\kern3.2pt\nointerlineskip}
  \ifContLineOne\hrulefill\else\vphantom{\hrulefill}\fi\crcr
  \noalign{\nointerlineskip}
  $\hfil\textstyle{\vbox to 14pt{}#1}\hfil$\crcr}}}

\def\DrawLeg#1#2{
  \kern-.2pt              
  \dimen2 =#1             
  \advance\dimen2 by 2pt  
  \dimen3 = 10.6pt        
  \dimen4 =3.6pt          
  \advance\dimen3 by -\dimen2 
  \multiply\dimen4 by #2
  \advance\dimen3 by \dimen4
  \raise\dimen2 \hbox{\vrule height\dimen3 width .4pt} 
  \kern-.2pt}             

\def\begC#1#2{\setbox0 =\hbox{$\textstyle{#2}$}
  \dimen0=.5\wd0 \dimen1=\ht0
  \conC{\hskip\dimen0}
  \count255=#1
  \ifnum\count255 =1 \ContLineOnetrue\else
  \ifnum\count255 =2 \ContLineTwotrue\else
  \ifnum\count255 =3 \ContLineThreetrue\fi\fi\fi
  \DrawLeg{\dimen1}{\count255}
  \conC{\hskip\dimen0}
  \kern-\dimen0\kern-\dimen0 \box0}

\def\endC#1#2{\setbox0 =\hbox{$\textstyle{#2}$}
  \dimen0=.5\wd0 \dimen1=\ht0
  \conC{\hskip\dimen0}
  \count255=#1
  \ifnum\count255 =1 \ContLineOnefalse\else
  \ifnum\count255 =2 \ContLineTwofalse\else
  \ifnum\count255 =3 \ContLineThreefalse\fi\fi\fi
  \DrawLeg{\dimen1}{\count255}
  \conC{\hskip\dimen0}
  \kern-\dimen0\kern-\dimen0 \box0}

\newtheorem{def1}{Definition}

\newtheorem{thm}{Theorem}
\newtheorem{lem}{Lemma}

\theoremstyle{definition}

\theoremstyle{remark}
\newtheorem{rem}{Remark}

\def\theequation{\arabic{section}.\arabic{equation}}
\def\eq#1{eq.~(\ref{#1})}
\def\Eq#1{Eq.~(\ref{#1})}

\def\eqs#1#2{eqs.~(\ref{#1}) and (\ref{#2})}

\def\Refs#1{refs.~\cite{#1}}
\newcommand{\bra}[1]{\langle #1|}
\newcommand{\ket}[1]{|#1\rangle}

\newcommand{\ie}{{\it i.e., }}
\newcommand{\eg}{{\it e.g., }}
\newcommand{\Lu}{{\cal L}}
\newcommand{\Ou}{{\cal O}}
\newcommand{\bea}{\begin{eqnarray}}
\newcommand{\eea}{\end{eqnarray}}
\newcommand{\nn}{\nonumber\\}

\newcommand{\osum}[1]{\vspace{0.5cm}{\pmb{\Circle}}{\hspace{-0.4cm}\sum_{#1}\hspace{0.1cm}} }

\usepackage{color}

\definecolor{orange}{rgb}{0.9,0.2,0}
\definecolor{brown}{rgb}{0.7,0.3,0.2}
\definecolor{fuxia}{rgb}{1,0,1}
\definecolor{skyblue}{rgb}{0,0.1,0.9}
\definecolor{violetred}{rgb}{0.8,0.13,0.56}
\definecolor{deeppink}{rgb}{1.00,0.08,0.5}
\definecolor{pink}{rgb}{1.00,0.75,0.80}
\definecolor{orchid}{rgb}{0.85,0.44,0.84}
\definecolor{lightpink}{rgb}{1.00,0.71,0.76}
\definecolor{bluish}{rgb}{0,0.6,0.8}

\title{\bf Mass Insertions vs.  Mass Eigenstates calculations in Flavour Physics}

\author{A.   Dedes$^{1,2}$\footnote{email: {\tt adedes@cc.uoi.gr}},~ 
M.   Paraskevas$^{1}$\footnote{email: {\tt mparask@grads.uoi.gr}},
J.   Rosiek$^{3}$\footnote{email: {\tt janusz.rosiek@fuw.edu.pl }},~
K.   Suxho$^{1}$\footnote{email: {\tt csoutzio@cc.uoi.gr}},~
K.   Tamvakis$^{1}$\footnote{email: {\tt tamvakis@uoi.gr}}}
\affil{\small $^{1}$Department of Physics, Division of Theoretical
  Physics, \\ University of Ioannina, GR 45110, Greece}
\affil{\small $^{2}$University of Athens, Physics Department, \\
Nuclear and Particle Physics Section, GR 15771 Athens, Greece}
\affil{\small $^{2}$Institute of Theoretical Physics, Physics
  Department, Warsaw University, \\ Pasteura 5, 02-093 Warsaw, Poland}

\date{April 7, 2015}

\begin{document}

\maketitle

\begin{abstract}

We present and prove a theorem of matrix analysis, the Flavour
Expansion Theorem (or FET), according to which, an analytic function
of a Hermitian matrix can be expanded polynomially in terms of its
off-diagonal elements with coefficients being the divided differences
of the analytic function and arguments the diagonal elements of the
Hermitian matrix.  The theorem is applicable in case of flavour
changing amplitudes.  At one-loop level this procedure is particularly
natural due to the observation that every loop function in the
Passarino-Veltman basis can be recursively expressed in terms of
divided differences.  FET helps to algebraically translate an
amplitude written in mass eigenbasis into flavour mass insertions,
without performing diagrammatic calculations in flavour basis.  As a
non-trivial application of FET up to a third order, we demonstrate its
use in calculating strong bounds on the real parts of flavour changing
mass insertions in the up- squark sector of the MSSM from neutron
Electric Dipole Moment (nEDM) measurements, assuming that CP-violation
arises only from the CKM matrix.

\end{abstract}

\newpage
\section{Introduction}
\setcounter{equation}{0}
\label{sec:intro}

Within the general framework of a perturbative Quantum Field Theory
(QFT), the standard strategy followed when calculating physical
transition amplitudes, is to express the Lagrangian density in a
particular field basis, commonly referred to as {\sl mass eigenstate}
basis.  Contrary to any other possible choice, only in this case the
states of the theory correspond to physical particles with definite
mass and symmetry charges.  Up to possible mass degeneracies, this
basis is unique and is characterized by the absence of quadratic
mixing terms between different mass eigenstates.  Furthermore, all
parameters of the Lagrangian in this basis, are physically observable,
in the sense that all masses and couplings can in principle be
obtained by a suitable experiment.  After having set the Lagrangian to
the mass eigenstates fields basis, one can then deploy the standard
QFT machinery and set the Feynman rules in order to calculate
transition amplitudes for any physical process.

Nevertheless, in the vast majority of the models we are interested in,
masses are typically generated or affected by a symmetry breaking
mechanism.  In this case another basis is physically meaningful as
well.  This is the basis where the Lagrangian exhibits explicitly the
properties of the initial symmetry, and the states correspond to
eigenstates of a larger symmetry group.  We will refer to these
eigenstates, for gauge bosons and collectively for fermions and
scalars with family replication, as {\sl gauge} and {\sl flavour
  eigenstates} respectively, although in our definition for the latter
there is no implicit requirement of an underlying flavour symmetry.
In this sense the flavour eigenstate basis in many models can be
considered in practice arbitrary, constrained only by the other
symmetries of the initial Lagrangian, {\sl i.e., gauge symmetry,
  supersymmetry, etc.}  The transformation from the initial basis to
the mass eigenstate basis, which is still the physical basis of the
theory, is performed with mass diagonalization involving unitary
transformations and field redefinitions, typically leaving a physical
imprint on the parameters of the mass eigenstate theory.  In the
Standard Model (SM)~\cite{Weinberg:1967tq,Glashow,Salam} this effect
is displayed in the gauge sector through the weak mixing angle and in
the fermion sector through the
CKM~\cite{Cabibbo:1963yz,Kobayashi:1973fv} and
PMNS~\cite{Pontecorvo:1957cp,Maki:1962mu} matrices of charged
currents.  However, even in this very successful model, the CKM or
PMNS parameters along with the fermion mass eigenvalues are
insufficient to determine unambiguously the flavour eigenstate basis.

Although the mass eigenstate basis of a perturbative QFT is the
natural basis for calculations of physical processes, some effects
typically related to the symmetries of the Lagrangian before symmetry
breaking are better understood in flavour basis.\footnote{Since in
  many cases the mass diagonalization of various sectors of the theory
  is independent of each other, one can also work in a mixed basis
  where some sectors are given in mass basis and others in flavour
  basis.  In what follows, the basis we work can be easily identified
  from the context.}  Therefore for a qualitative analysis of such
effects, it is often useful to have our expressions in the latter
basis.\footnote{This is after all the basis that more naturally
  connects couplings and masses to high energies through their
  Renormalization Group Equations (RGEs).  }  The standard strategy
that has been employed up to date, is an approximate diagrammatic
method commonly referred to as the {\sl Mass Insertion Approximation}
(MIA)\cite{Gabbiani:1996hi, Misiak:1997ei}.  In this approach one
defines the Feynman rules of the theory directly in flavour basis.
The diagonal part of the flavour mass matrix is typically absorbed
into the definition of (unphysical) massive propagators and the
non-diagonal part commonly referred to as {\sl mass insertions} is
treated perturbatively, as part of the interaction Lagrangian which
now possesses quadratic mixing terms.  Due to the presence of
quadratic interactions, besides the standard loop approximation of a
perturbative QFT, there is an extra approximation for each diagram,
appearing as an infinite series in terms of mass insertions, following
the presence of a flavour propagator.

In what follows, we present an {\sl algebraic} treatment of transition
amplitudes in mass eigenstate basis, leading directly to the
corresponding amplitudes in flavour basis, in the form of the MIA or
of an equivalent expansion.
In particular, we prove a theorem in matrix analysis~\cite{bhatia97,
  0521386322}, that we coin {\sl Flavour Expansion Theorem} or simply
{\sl FET}, which says that an analytic function of a Hermitian matrix
can be expanded polynomially in terms of its off-diagonal elements
with coefficients being the divided difference of the analytic
function and arguments the diagonal elements of the Hermitian matrix.
At one-loop level, this expansion is naturally related to the
remarkable recursive properties of next order Passarino-Veltman (PV)
function~\cite{Passarino:1978jh} being the {\sl divided
  difference}~\cite{de2005divided} of the previous one.
We then argue that FET connects mass and flavour field bases
amplitudes.
The first non-trivial order in the expansion [{\it cf.} \eq{theorem}],
and applications in MSSM flavour physics, have been presented
in~\Refs{Buras:1997ij,Giudice:2008uk} but a formal mathematical proof
to all orders was unknown until now.
FET is especially useful when is used to evaluate higher order
expansion terms, it is technically easier, elegant and superior to
often tedious, time-consuming and thus prone to mistakes calculations
of the diagrammatic MIA.  We support our claims with a novel example
towards the end of the article.

More specifically, the paper is organized as follows: In
Section~\ref{sec:calcs}, we present a warming-up example of a scalar
toy-model in order to illustrate the relation between the calculation
of flavour transition amplitudes in mass and flavour bases.  Then, in
Section~\ref{sec:fet}, we formulate a general algebraic theorem
dealing with the expansion of an analytic function of a Hermitian
matrix, and, discuss its applications to rewriting flavour amplitudes
with scalar and vector particles, from mass to flavour eigenstates
basis.  We extend our discussion to the case of amplitudes involving
fermions in Section~\ref{sec:ferexp}.  In Section~\ref{sec:nedm}, we
illustrate the developed technique on a (potentially) physical
example, expanding the dominant gluino-squark contribution to the
neutron Electric Dipole Moment and showing the importance of higher
order terms.  We conclude our results in Section~\ref{sec:summary}.
Finally, the formal proof of theorem formulated in
Section~\ref{sec:fet} is given in \ref{app:fet}, while in
\ref{app:pvconv} we derive the convergence criterion for the mass
insertion expansion of the one-loop integrals.

\section{A warming-up example: flavour calculation techniques}
\setcounter{equation}{0}
\label{sec:calcs}

To set up a simple framework to introduce the standard techniques of
flavour physics calculations, we consider a scalar field toy model
composed of $N$-complex charged scalar fields $\Phi_I$, with family
replication, and an extra neutral, real, scalar field $\eta$.  The
(squared) mass matrix, ${\bf M}^2$, and the Yukawa coupling matrix,
${\bf Y}$, of the flavour eigenstates $\Phi_I$, are necessarily
Hermitian but not aligned in general.  The Lagrangian density, will
have the form:\footnote{Sum over repeated indices is always assumed in
  the text, unless stated otherwise.}
\bea
\Lu_{\rm flavour} = (\partial^\mu \Phi_I^\dag )\,(\partial_\mu \Phi_I)
-M_{IJ}^2 \Phi_I^\dag \Phi_J + \,\frac{1}{2}(\partial^{\mu}\eta) \,
(\partial_{\mu}\eta) - \frac{1}{2} m_\eta^2 \eta^2 \, - Y_{IJ}\;\eta\,
\Phi_I^\dag\, \Phi_J - \ldots \;,
\label{Luflavour}
\eea
where dots denote additional scalar field interactions which are
irrelevant for the discussion below.  Using the unitary rotation,
\begin{equation}
\Phi_I =  U_{Ii}\;\phi_i\,,\,
\end{equation}
where $\mathbf{U}$ satisfies the condition
\begin{equation}
{\bf U^\dagger \, M^2 \, U} = {\bf m^2} = {\bf diag}(m_1^2,\ldots,m_N^2) \;,
\end{equation}
one can express the Lagrangian in terms of mass eigenstates $\phi_i$
\bea
\Lu_{\rm mass} = (\partial^\mu \phi_i^\dag )\,(\partial_\mu \phi_i) -
m_{i}^2\phi_i^\dag \phi_i + \,\frac{1}{2}(\partial^{\mu}\eta) \,
(\partial_{\mu}\eta) - \frac{1}{2}m_\eta^2 \eta^2
- y_{ij}\;\eta\, \phi_i^\dag\,\phi_j + \ldots \;,
\label{Lumass}
\eea
where the transformed scalar couplings are identified as $y_{ij} =
U_{iI}^{\dagger}\, Y_{IJ}\, U_{Jj}$.

\begin{figure}[t]
{\center{ \includegraphics[trim=0.5cm 4.cm 2cm 2.7cm, clip=true,
      totalheight=0.22\textwidth]{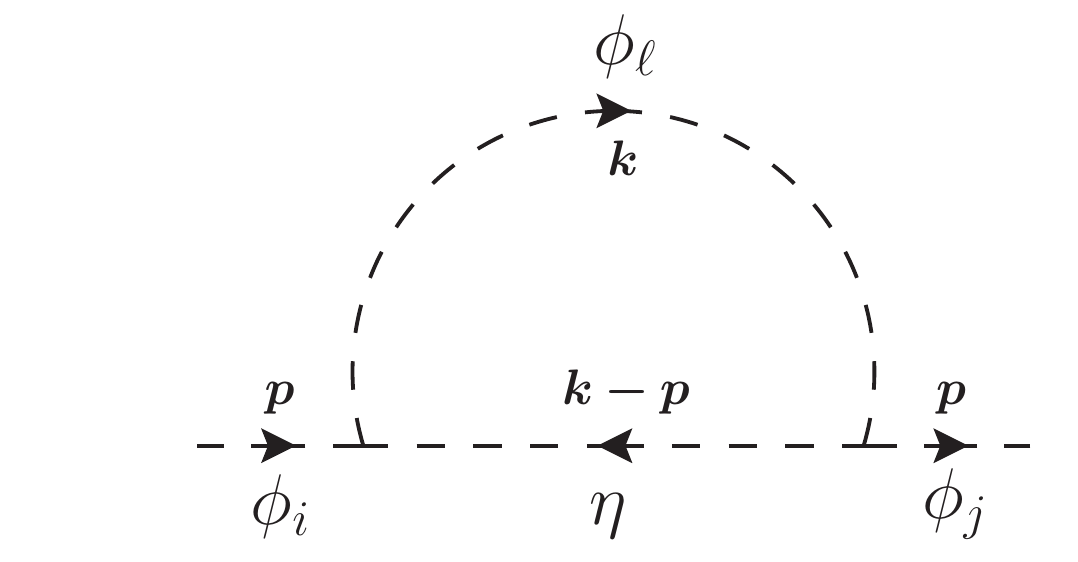}}
    \caption{Scalar self-energy \textit{$-i\Sigma_{ji}(p)$ } in the
      mass eigenstate basis.  }
    \label{fig:eigense}}    
\end{figure}
\begin{figure}[t]
\center{\includegraphics[trim=0.cm 5.4cm 0cm 1.cm, clip=true,
    totalheight=0.26\textwidth]{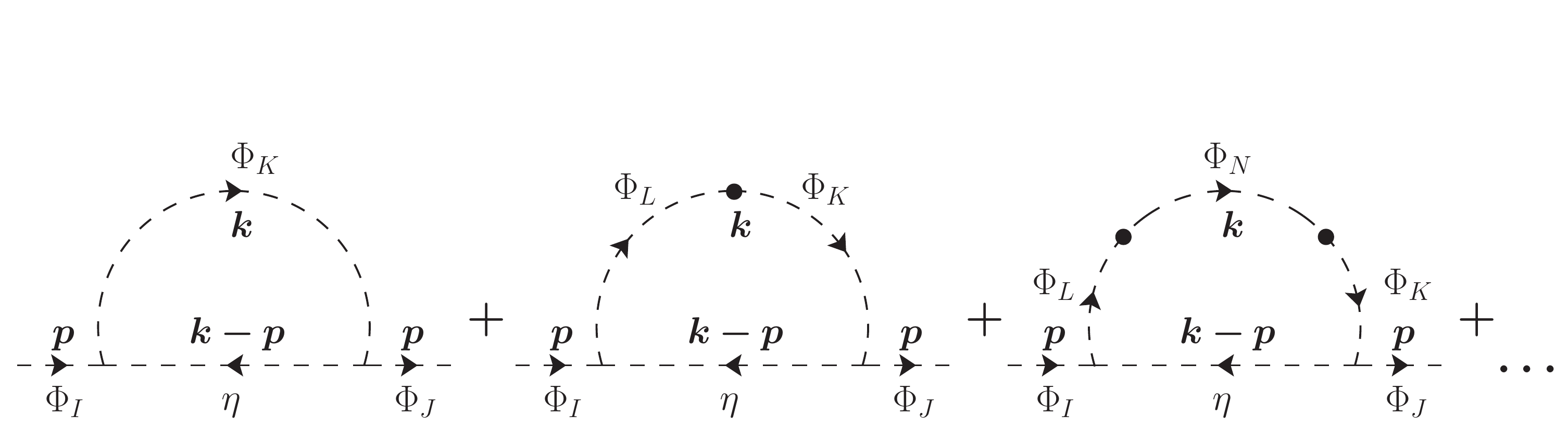} }
\caption{Scalar self-energy \textit{$-i\hat{\Sigma}_{JI}(p)$ } in the
  flavour basis.
\label{fig:miase}  } 
\end{figure}

First, let us consider the ``flavour'' changing one-loop
One-Particle-Irreducible (1PI) self-energy diagram of the mass
eigenstates fields $\phi_i$, shown in Fig.~\ref{fig:eigense}.  Then
using interactions from Lagrangian in \eq{Lumass} leads to
\begin{equation}
-i \Sigma_{ji}(p)\,=\, \int\frac{d^4k}{(2\pi)^4}\: y_{j\ell}\:
\frac{1}{(k^2 - m_\ell^2)((k-p)^2 - m_\eta^2 )} \: y_{\ell i} \ =
\ \frac{i}{(4\pi)^2} \: y_{j\ell} \: B_0(p;m_\ell^2,m_\eta^2) \:
y_{\ell i} \;,
\label{Sigmaji}
\end{equation}
where the loop function $B_{0}$ (and also $C_{0},D_{0},...$ below) is
a PV-function defined in Section~\ref{sec:pvfun}.

Next, we consider the corresponding one-loop diagram in flavour basis
of \eq{Luflavour}, employing the standard diagrammatic MIA approach.
The massive flavour propagators are defined by absorbing the diagonal
part of the flavour mass matrix ${\bf M}^2$, according to the
decomposition into diagonal and non-diagonal parts,
\bea
M^2_{IJ} \ = \ (M_0^2)_{II} \: \delta_{IJ} +  \hat{M}^2_{IJ}\,,\qquad
\hat{M}^2_{II}=0\, ,\;\emph{ (no sum over I)}\label{M^2} \;,
\eea 
where $\delta_{IJ}$ is the usual Kronecker symmetric tensor.  The
non-diagonal elements $\hat{ M}^2_{IJ}$, are identified as the mass
insertions, treated as perturbative couplings for the non-diagonal
quadratic interactions of flavour fields.  The one-loop flavour
changing self-energy of the flavour states, $\Phi_I$, is thus
represented as the infinite sum of the diagrams shown in
Fig~\ref{fig:miase}.  By direct calculation, and denoting $M^2_K\equiv
(M_0^2)_{KK}$, one obtains
\bea
-i \hat{\Sigma}_{JI}(p) & = &\int\frac{d^4k}{(2\pi)^4} \: { Y_{JK}
  Y_{LI} \over (k-p)^2 - m_\eta^2}\times \Bigg({\delta_{KL} \over k^2
  - M_K^2}\ +\ {\hat{M}_{KL}^2 \over (k^2 - M_K^2)(k^2 - M_L^2)}
\nonumber\\
&+& {\hat{M}_{KN}^2 \hat{M}_{NL}^2 \over (k^2 - M_K^2)(k^2 -
  M_N^2)(k^2 - M_L^2)} \ +\ \ldots\Bigg)\label{eq:miase0}\\
& = &\, \frac{i}{(4\pi^2)} \: Y_{JK} Y_{LI} \times \Bigg(\delta_{KL}
B_0(p;M_K^2,m_\eta^2)\ +\ \hat{M}_{KL}^2 C_0(0,p;M_K^2,M_L^2,m_\eta^2)
\nonumber\\
&+& \hat{M}_{KN}^2\hat{M}_{NL}^2 D_0(0,0,p;M_K^2,M_N^2,M_L^2,m_\eta^2)
\ +\ \ldots\Bigg)\;,
\label{eq:miase}
\eea
which is essentially an expansion in terms of mass insertions.  We
should notice that although $M_K^2$ are not the squares of the
physical masses, they are always real non-negative.  This is a general
property of the diagonal part of any semi-positive definite Hermitian
matrix, including also any Hermitian (squared) mass matrix of a
consistent QFT.

As the external indices imply, the self energy diagrams in this
example are not invariant under flavour rotations.  One can formally
uncover the explicit correspondence between the flavour and mass basis
calculations by considering the flavour invariance of the time
evolution operator, inside the corresponding $S$-scattering matrix
element, with the relevant contractions
\bea 
\int \hspace{-0.15cm}d^4x \hspace{-0.15cm}\int \hspace{-0.15cm}
d^4y\; \begC1{\bra{p,J}} \endC1{\Phi}_{J}^\dag (x) \hat{\Ou}^{JI}_{(x,y)}
\begC1{\Phi_{I}(y)}\endC1{\ket{p',I}} &=&
\int \hspace{-0.15cm}d^4x \hspace{-0.15cm}\int \hspace{-0.15cm}
d^4y\;\begC1{\bra{p,J}} \endC1{\phi}_{j}^\dag (x){\Ou}^{ji}_{(x,y)}
\begC1{\phi_{i}(y)}\endC1{\ket{p',I}} \;,\nn
\begC1{\bra{p,J}}\endC1{\phi}_{j}^\dag(x) = \bra{0}
e^{ipx}\;U_{Jj} &,& \begC1{\phi_{i}(y)}\endC1{\ket{p',I}} =
U^\dag_{iI}e^{-ip'y}\ket{0}\;,
 \label{JIelement}
\eea 
we derive the transformation rule for self-energies,
\bea
 \hat{\Sigma}_{JI}(p) = U_{Jj}\; \Sigma_{ji}(p) \;U^\dag_{iI}\label{SigmaU}\;,
\eea 
which can be immediately generalized to the case of an arbitrary
$n$-point amplitude.

Substituting in \eq{SigmaU} the explicit algebraic expressions for the
self-energies we obtain an interesting result - the flavour rotation
of the mass eigenstates loop-function is an expansion in terms of mass
insertions in flavour basis (no sum over $K,L$),
\bea
U_{K\ell}\; B_0(p,m_\ell^2,m_\eta^2)\; U^\dag_{\ell L}& =& \delta_{KL}
B_0(p;M_K^2,m_\eta^2)\ +\ \hat{M}_{KL}^2 C_0(0,p;M_K^2,M_L^2,m_\eta^2)
\nn
&+& \hat{M}_{KN}^2\hat{M}_{NL}^2 D_0(0,0,p;M_K^2,M_N^2,M_L^2,m_\eta^2)
\ + \ldots \;.
\label{eq:B0U}
\eea
This result, however, can be also obtained by a theorem of matrix
analysis [{\it cf.}  \eq{theorem}] stated in the next section,
rendering diagrammatic calculations in flavour basis, similar to ones
leading to \eq{eq:miase}, obsolete.

\section{Flavour Expansion Theorem}
\setcounter{equation}{0}
\label{sec:fet}

\Eq{eq:B0U} has been obtained diagrammatically with the help of the
transformation rule \eqref{SigmaU}.  In what follows we show that,
such relations can be also obtained purely algebraically, allowing for
an easier transformation between mass and flavour basis calculations
without the use of the diagrammatic MIA.  In this section we formulate
a relevant mathematical framework and a useful general theorem of
matrix analysis.  For brevity we refer to it as ``Flavour Expansion
Theorem'' or just FET.

Before formulating FET, it is worth noting that obtaining the
relation~(\ref{eq:B0U}) in a closed form without reverting to
diagrammatic MIA expansion is not easy with the use of standard
perturbation techniques.  The simplest idea of expanding the mass
eigenstates result in a Taylor series around some average mass $m_0^2
= \frac{1}{N}\sum_{K=1}^N m_K^2$,
\bea
U_{iK} \, f(m_K^2) \, U_{jK}^{\star} &=& U_{iK} \, \left[\,
\sum_{n=0}^\infty \frac{f^{(n)}(m_0^2)}{n!} \, (m_K^2 - 
m_0^2)^n \right]\, U_{jK}^{\star}\, \,,
\label{eq:eigexp}
\eea
fails to reproduce correctly even the first non-trivial MIA term
in~\eq{eq:B0U} - the higher terms of {\em any} order in Taylor
expansion are proportional to a factor
\bea 
U_{iK}\, (m_K^2 - m_0^2)^n \,
U_{jK}^{\star} \,=\, \left[ \left(\mathbf{M^2} -
m_0^2\mathbf{I}\right)^{n}\right]_{ij} \,=\, \left[ \left(\mathbf{M_0^2}
- m_0^2\mathbf{I} + \mathbf{\hat M^2}\right)^n\right]_{ij} \;,
\eea
which after expanding would explicitly contain terms linear in $\hat
M^2_{ij}$.

Alternatively, a more consistent approach can be developed using the
standard quantum mechanic perturbation theory, applied to mass matrix
eigenstate problem,
\bea
\left(\mathbf{M_0^2} + \lambda \mathbf{\hat M^2}\right) \, \left(\mathbf{e_I^{(0)}} +
\lambda \mathbf{e_I^{(1)}} + \ldots\right) \ = \ \left[(m^2)^{(0)}_I + \lambda
(m^2)^{(1)}_I + \ldots \right]\, \left(\mathbf{e_I^{(0)}}
+ \lambda \mathbf{e_I^{(1)}} +
\ldots \right)\;,
\label{eq:qmpert}
\eea
with $\lambda$ being the expansion parameter.  By solving
\eq{eq:qmpert} order by order we find mass eigenvalues $m_I^2$ and
rotation matrices $\mathbf{U}=(\mathbf{e_1},\ldots,\mathbf{e_N})$ as a
series in $\lambda$.  Then the product, $U_{iK}\, f(m_K^2)
\,U_{fK}^{\star}$, can again be expanded in Taylor series in $\lambda$
parameter, with each term equivalent to the same order of MIA
expansion.  Such procedure easily restores the first terms
in~\eq{eq:B0U} (see e.g.~\cite{Crivellin:2010gw}), but its
combinatorial complication grows quickly and again it is hard to see
how the higher order terms combine to form compact $n$-point loop
functions, a situation which becomes even trickier in case of
degenerate eigenvalues.

\subsection{Hermitian matrix function and Divided Differences}
\label{sec:math}

We append below definitions that are relevant for presenting the
Flavour Expansion Theorem.

\begin{def1}
{Let $\mathbf{A}$ be an $n\times n$ Hermitian matrix, diagonalized by
  a unitary matrix $\mathbf{U}$ to a real diagonal matrix
  $\mathbf{D}$, through the transformation,
\bea
\mathbf{U^\dag A U} = \mathbf{D} =
\mathbf{diag}(\lambda_1,\dots,\lambda_n)\, .
\eea
Also let $f(x)$ be a real analytic function about zero, in a range
$I\subseteq \mathbb{R}$, that can be expressed in terms of its
Maclaurin series, as
\bea
f(x)=\sum_{m=0}^{\infty}c_m x^m\;.
\eea 
Then, if all $\lambda_i \in I$, one can define a Hermitian matrix
function $f({\bf A})$, as
\bea
\mathbf{U} f(\mathbf{D}) \mathbf{U^\dag}\ = \ \sum_{m=0}^{\infty}
\ c_m \mathbf{U} \mathbf{D}^m \mathbf{U^\dag} \ = \ \sum c_m
\mathbf{A}^m \equiv f (\mathbf{A})\ .
\label{eq:hermfun}
\eea
\label{def1}}
\end{def1}

\begin{def1}
For any function $f(x)$, one can define a set of multi-variable
functions, $f^{[k]}$, through a mathematical operation commonly
referred to as divided difference.  Divided differences are defined
recursively as
\begin{subequations}
\begin{align}
f^{[0]}(x)&\equiv f(x)\: \;,\\
f^{[1]}(x_0,x_1)&\equiv\frac{f(x_0)-f(x_1)}{x_0-x_1}\: \;, \\
\ldots\nonumber\\
f^{[k+1]}(x_0,\dots ,x_k,x_{k+1})&\equiv\frac{f^{[k]}(x_0,\dots ,
  ,x_{k-1},x_k)-f^{[k]}(x_0,\dots ,x_{k-1},x_{k+1})}{x_k-x_{k+1}} \;.
\end{align}
\label{dddef}
\end{subequations}
\label{def2}
\end{def1}
It is easy to check that divided differences of any order $k$ are
always totally symmetric under the permutation of any set of
respective arguments, $x_i$.  Moreover, for an analytic generating
functional $f=f^{[0]}$ they also have a well defined degeneracy limit
\begin{align}
\lim_{\{ x_0,\dots ,x_m\} \to \{ \xi,\dots,\xi \}}
f^{[k]}(x_0,\dots,x_k)=\frac{1}{m!}\frac{\partial^{m} }{\partial
  \xi^{m}}f^{[k-m]}(\xi,x_{m+1}\dots,x_k)\;,
\label{ddlimit}
\end{align}
applied to any set of respective arguments ($m\leq k$), as long as all
arguments lie in the domain of analyticity of $f(x)$.

\subsection{The Flavour Expansion Theorem (FET)}
\label{sec:fetdef}

By making extensive use of the definitions above, we can formulate a
general theorem, concerning a certain expansion of Hermitian matrix
functions, in a form which can be directly applied to the unitary
transformation of loop functions in flavour physics.
\begin{thm}
{Let $\bf{A}$ be any $n\times n$ Hermitian matrix, decomposed as a sum
  of its diagonal and non-diagonal part,}
\begin{equation}
\mathbf{A}\,=\,\mathbf{A_0}\,+\,\mathbf{\hat{A}}\, ,\label{Ageneral}
\end{equation}
{{where, by definition,}}
\begin{equation}
\begin{array}{l}
A_0^I\,\equiv A_{II}\;,\\[3mm]
\hat{A}_{IJ}\,\equiv\,A_{IJ},\qquad\hat{A}_{II}=0\;, \hspace{1cm}
(I,J=1,\dots,n)\;.
\end{array}\label{Adiag}
\end{equation}
Then, for any Hermitian matrix function $f({\bf {A}})$, satisfying the
  restrictions of Def.~\ref{def1}, a given matrix element $\bra I
  f(A) \ket J $ will be given by the expansion (no sum over I,J)
\begin{align}
f({\bf A})_{IJ} &= \delta_{IJ}f(A^I_0)\ +
f^{[1]}(A_0^I,A_0^J)\:\hat{A}_{IJ}+ \sum_{K_1}
f^{[2]}(A_0^I,A_0^J,A_0^{K_1})\:
\hat{A}_{I{K_1}}\hat{A}_{{K_1}J}\nonumber\\
& + \sum_{K_1,K_2}f^{[3]}(A_0^I,A_0^J,A_0^{K_1},A_0^{K_2})\:
\hat{A}_{I{K_1}}\hat{A}_{{K_1}{K_2}}\hat{A}_{{K_2}J} + \dots \;\;,
\label{theorem}
\end{align}
in terms of divided differences of $f(\bf{A_0})$ [see Def.~\ref{def2}]
  and the (non-diagonal) elements of~$\bf{\hat{A}}$.\footnote{Cases of
  degenerate eigenvalues and/or diagonal matrix elements are treated
  uniformly due to property \eqref{ddlimit} of the divided
  differences.}
\end{thm}
Eq.~\ref{theorem} holds as long as the expansion in the RHS is
convergent.  The formal proof of this theorem, based on the notion of
fully symmetrized polynomials and mathematical induction, is given
in~\ref{app:fet}.

\subsection{Divided Differences and Passarino-Veltman  functions}
\label{sec:pvfun}

The natural connection between FET and the expansion of one-loop
amplitudes is becoming striking when looking into the recursive
properties of the loop functions in the Passarino-Veltman
basis~\cite{Passarino:1978jh}.  The general $n$-point one-loop
functions can be defined as:
\begin{align}
&\frac{i}{(4\pi)^2}PV_n^{\mu_1\ldots\mu_l}(p_1, \dots,
    p_{n-1};m_1^2, \dots,m_{n}^2) = \nonumber \\[2mm]
&\hskip 30mm \int \frac{d^4k}{(2\pi)^4}\frac{k^{\mu_1}\ldots
    k^{\mu_l}}{(k^2-m_1^2)\;\prod_{j=2}^{n} ((k + p_1 + \dots +
    p_{j-1})^2-m_{j}^2)}\;, \quad(n\geq 2)\;.
\label{PVn}
\end{align}
In the standard notation $n=2,3\ldots$ functions are commonly denoted
as $B,C,\ldots$--functions.

A useful property associates differences of integral functions of a
certain order with integral functions of next order.  In general case
this relation has the following structure:
\bea
&& {PV_n^X(p_1\dots p_{n-1}; m_1^2\dots m_{n}^2) - PV_n^X(p_{1}\dots
     p_{n-1}; {m'_1}^2\dots m_{n}^2) \over m_1^2-{m'_1}^2} \nonumber\\
& & \hskip 55mm = PV_{n+1}^X(0,p_{1}\dots p_{n-1}; m_1^2,{m'_1}^2\dots
m_{n}^2) \;, \label{lgeq2}\\[4mm]
&& {PV_n^X(\dots p_{j-1}\dots;\dots m_j^2\dots ) - PV_n^X(\dots
     p_{j-1}\dots;\dots {m'_j}^2\dots) \over m_j^2-{m'_j}^2}
\nonumber\\
&& \hskip 55mm = PV_{n+1}^X(\dots p_{j-1},0\dots ;\dots
m_j^2,{m'_j}^2\dots ),\;\;\; (j\geq 2)\nonumber
\eea
with $X$ being any set of Lorentz indices of momenta in the numerator
of loop integrand.\footnote{Additional discussion and more recursive
  relations for the various types of PV functions can be found in
  Appendix A of ref.~\cite{Dedes:2014asa}.}

Comparing \eq{lgeq2} with the definition~(\ref{dddef}) one can see
immediately that the notion of divided differences is \emph{naturally}
implemented in the relations between multi-point one-loop integrals.
Eq.~(\ref{lgeq2}) allow us to express the expansion~\eqref{theorem}
for one-loop amplitudes in a form in which it is obvious that it is
not singular for degenerate diagonal matrix elements.  Namely, every
one-loop amplitude can be written as a linear combination of PV
functions.  Furthermore, each PV function can be expanded as,
\bea
\left[\:PV^{(n)}(\ldots,A,\ldots)\:\right]_{IJ} & = &
\delta_{IJ}PV^{(n)}(\ldots,A^I_0,\ldots)\ +
\ PV^{(n+1)}(\ldots,A^I_0,A_0^J,\ldots)\hat{A}_{IJ} \nonumber \\
& + & \sum_{K}PV^{(n+2)}(\ldots,A^I_0,A_0^J,A_0^K,\ldots)\hat{A}_{IK}
\hat{A}_{KJ}\ + \ \dots \;,
\label{theoremPV}
\eea
where if necessary one should also specify momenta arguments as
prescribed in \eq{lgeq2}.

For example, to make a connection between FET and the toy-model of the
previous section we observe that if we make the following
identifications
\bea
\quad {\bf{D}} \rightarrow {\bf{m}^2}\;,\quad {\bf{A}} \rightarrow
      {\bf{M}^2}\;,\quad f(x)\equiv f^{[0]}(x) \rightarrow
      B_0(p,x,m_\eta^2)\;,
\eea
we can immediately see that~\eq{eq:B0U} is a special case
of~\eq{theorem}.  In Section~\ref{sec:nedm}, we will present a highly
non-trivial example application of the FET.

\subsection{Applications and limitations of FET}
\label{sec:fetlimits}

The FET formulated as a pure algebraic theorem can be directly applied
to expanding a mass eigenstate result of any transition amplitude in
any model involving particles associated with Hermitian mass matrices,
that is scalars or vectors, even at higher loop orders.  As we shall
discuss in Section~\ref{sec:ferexp}, with some modifications, FET can
be also applied to expanding the amplitudes involving fermions
associated with non-Hermitian mass matrices.

The purely algebraic expansion is usually significantly simpler than
the more tedious and prone to mistakes diagrammatic MIA calculation,
particularly in models with complicated flavour structure like MSSM.
Another advantage of FET is that it can be easily implemented as an
algorithm for symbolic manipulation programs, automatizing the
expansion procedure.  However, the procedure has some limitations,
particularly when is applied to such complicated functions as loop
integrals.  Three remarks concerning limitations of FET are summarized
here:

\begin{rem} FET assumptions require the expanded amplitude to be analytic
  function of masses.  This is not the case if external momenta are
  large and loop integrals may have branch cuts.  For the external
  momenta in the vicinity of branch cuts a mass eigenstate calculation
  and use of numerical procedures is more appropriate.
\end{rem}
\begin{rem} Flavour expansion in the r.h.s of \eq{theorem} 
 may not converge or converge very slowly, in both cases again mass
 eigenstates basis and use of numerical procedures is preferred.
\end{rem}
\begin{rem} The UV-singularities do not pose a problem for FET.  If they
  appear, they come from the loop integrals of positive mass
  dimension.  The coefficients of poles of such integrals are
  dimensionless and flavour blind or proportional to positive powers
  of masses, so they can be evaluated in terms of flavour basis
  parameters without any expansion.
\end{rem}

Most practical applications of FET concern analyses of models of New
Physics where loop particles are much heavier than the external states
(being usually the Standard Model fields).  Thus, it is usually
sufficient to calculate relevant amplitudes in the approximation of
vanishing external momenta, or by expanding the loop integrals in the
external momenta and keeping only the first few terms of such an
expansion.  Since in these processes the loop integrals are real
analytic functions of masses, branch cuts can never appear.  Then, the
only remaining problem is the convergence of the FET.

In~\ref{app:pvconv}, we formulate and prove the condition which has to
be fulfilled by the mass matrices in the flavour basis, in order to
make FET written for any one-loop amplitude, convergent.  The result
is that the moduli of every eigenvalue of the dimensionless mass
insertion matrix has to be smaller than one.

\section{Expansion of  fermionic amplitudes}
\setcounter{equation}{0}
\label{sec:ferexp}

Expanding amplitudes in which flavour violation enters through
fermionic mass matrices is more complicated, because such matrices do
not need to be Hermitian and in general can be diagonalized with the
use of two different unitary matrices.  Nevertheless, as it turns out,
FET can always apply to this case as well, with minor, but necessary,
modifications which we discuss below.
 
Lets first consider a Lagrangian of $N$-copies of Dirac fermion free
fields.  This will have the general form,
\bea
\Lu^{(0)}_{\rm flavour} &=&i \bar\Psi_A\slashed{\partial} \Psi_A -
M_{AB} \big(\bar\Psi_A P_L\Psi_{B}\big) - M^\dag_{AB} \big(\bar\Psi_A
P_R \Psi_{B}\big)\label{eq:Lfermion}\\
& \equiv & \mathbf{\bar\Psi}\Big(i\slashed{\partial} -\mathbf{M} P_L -
\mathbf{M^\dag} P_R\Big) \mathbf{\Psi} \;,\label{eq:LfermionMat}
\eea
in a 4-spinor Dirac notation, which is more suitable for mass
eigenstates calculations and offers a more compact description in our
following discussion.  Since Majorana spinors can be understood as
Dirac spinors with an extra chirality constraint, our discussion
applies directly to Majorana fermions, as well.

As is well known, in a chiral theory, a Dirac spinor is in general
reducible under flavour rotations.  The transformation from flavour to
mass basis is performed through two different unitary matrices, acting
independently on its chiral projections, as
\bea
{\Psi_{L}}_A = U_{Ai}\psi_{Li}\;, \qquad {\Psi_{R}}_A =
V_{Ai}\psi_{Ri}\label{UpsiLR} \;,
\eea
which can always bring an arbitrary complex flavour mass matrix ${\bf
  M}$ into a real non-negative diagonal form, satisfying
\bea
\mathbf{V^\dagger\, M\, U}\ = \ \mathbf{m} \ = \ \mathbf{diag}
(m_1,\ldots, m_N)\,.
\eea
The unitary matrices $\mathbf{V}$ and $\mathbf{U}$ diagonalize also
the (semi) positive-definite Hermitian matrices ${\bf MM^\dag}$ and
${\bf M^\dag M}$, through the transformations
\bea
{\bf V^\dagger\, M\: M^\dagger\, V} \ = \ {\bf U^\dagger\, M^\dagger\: M\, U} \ = \ {\bf
  m^2}\,.   
\label{eq:uvdef}
\eea

To streamline the notation, we introduce the unitary matrices
$\bm{{\cal U}}$ and $\bm{{\cal \bar U}}$, generalizing our
transformation rules for chirality projected fermion fields to a
reducible Dirac 4-spinor, as
\bea 
\bm{{\cal U}} \equiv \mathbf{U} P_L + \mathbf{V} P_R \;,\qquad
 \bm{{\cal \bar U}} \equiv \mathbf{U^\dag} P_R +
\mathbf{V^\dag} P_L\;.
\label{Udirac}
\eea
In this compact description, \eq{UpsiLR} will result in
\bea 
\Psi_A = {\cal U}_{Ai}\psi_i,\qquad \bar\Psi_A = \bar\psi_i\;
\bar{\cal U}_{iA}\;.
\label{Udiracf}
\eea
The free propagator for the fermion multiplet ${\bf\Psi}$ in flavour
basis is a matrix both in spinor and flavour space.  Inverting the
Dirac operator in~\eq{eq:LfermionMat}, we find
\bea
\mathbf{\hat\Delta}(k) &=& \frac{i}{\slashed k -\mathbf{M} P_L -
  \mathbf{M^\dag} P_R} \nonumber \\[2mm]
&=& (\mathbf{M^\dag}P_L+\slashed{k} P_L) \frac{i}{ k^2 -\mathbf{M
    M^\dag}}+(\mathbf{M}P_R + \slashed{k} P_R) \frac{i}{ k^2
  -\mathbf{M^\dag M}}\label{mat-prop}\, .
\eea

The free propagators in flavour and mass eigenstates basis are related
by the same rotations as fermion fields.  From the identity,
\bea 
\bra{0}T\{\Psi_B(x) \bar\Psi_A(y)\}\ket{0} &=& {\cal U}_{Bi}\,
\bra{0}T\{\psi_i(x)\bar\psi_i(y)\} \ket{0} \, \bar{\cal U}_{iA} \;,
\eea
it follows that, 
\bea 
\Big({\bm{\hat{\Delta}}}(k) \Big)_{BA} = \Big({\bm{ \mathcal
    U}}\;\bm{\Delta}(k)\;{\bm{\bar{\mathcal{U}}}}\Big)_{BA}&=&\left(
\;{\bm{ \mathcal U}}\; {i\, (\slashed{k} +  \mathbf{m}) \over k^2-{\bf
    m^2}}{\;\bm{\bar{\mathcal{U}}}} \; \right)_{BA} \;,
\label{eq:proptrans}
\eea
where $\Delta_{i}(k)$ is the fermion propagator in the mass
eigenstates basis.  Applying the explicit expressions of \eq{Udirac}
for the reducible flavour rotation matrices, to \eq{eq:proptrans} and
using the following algebraic identities:
\bea
\bar{\bm{\mathcal
    U}}^\dag\; {1 \over k^2 - \mathbf{m^2}}\;
\bar{\bm{ \mathcal
    U}} &=& \frac{1}{k^2 -\mathbf{M M^\dag}} P_L +
\frac{1}{ k^2 -\mathbf{M^\dag M}}P_R \;,  \\[2mm]
\bm{ \mathcal
    U}\;{\mathbf{m} \over k^2 - \mathbf{m^2}} \;
\bar{\bm{ \mathcal
    U}} &=& \mathbf{M^\dag} \frac{1}{ k^2 -\mathbf{M
    M^\dag}}P_L+\mathbf{M} \frac{1}{ k^2 -\mathbf{M^\dag M}}P_R \; ,
\label{propdecomp}
\eea
the flavour propagator $\bm{\hat{\Delta}}(k)$ can be also obtained
from the mass basis propagator $\bm{\Delta}(k)$.

In order to calculate the amplitude, apart from propagators one needs
to consider the transformation rules for the fermionic vertices.  A
general fermionic current in flavour basis, can be expressed in the
form,
\bea
j_\Psi = \bar\Psi_A \hat\Gamma_{AB} \Psi_B \;,
\eea
where $\mathbf{\hat\Gamma}$ is an operator acting both in flavour and
spinor space, and may also depend on scalar or gauge fields.  In any
QFT model a general fermionic vertex can be decomposed into four
chirality projected terms as
\bea
\bm{\hat \Gamma} & = & \, \bm{\hat \Gamma}_{RL} \,P_L + \,\bm{\hat
  \Gamma}_{LR} \,P_R + \, \bm{\hat \Gamma}_{RR}\, P_R + \,\bm{\hat
  \Gamma}_{LL} \,P_L \; \nn
&\equiv & \,P_L\,\bm{\hat \Gamma}_{RL} + P_R\,\bm{\hat \Gamma}_{LR} \,
+ \, P_L\, \bm{\hat \Gamma}_{RR}\, + P_R \,\bm{\hat \Gamma}_{LL} \,
\;.\label{eq:gammadecomp}
\eea
where $\bm{\hat\Gamma}_{LR(RL)}$ are scalar- or tensor-type couplings,
and $\bm{\hat\Gamma}_{LL(RR)}$ are vector couplings.

In our compact notation, the transformation rule for vertices can be
simply expressed as:
\bea
\mathbf{\hat\Gamma} = \bm{\bar{\mathcal{U}}}^\dag \; \mathbf{\Gamma}
\; \bm{{\mathcal{U}}}^\dag \;,
\label{eq:gamtrans}
\eea
or explicitly in terms of $\bm{\hat{\Gamma}}$ components as
\begin{align}
& \mathbf{\hat\Gamma}_{RL} = \mathbf{V}\; \mathbf{\Gamma}_{RL} \;
  \mathbf{U}^\dag \;, \qquad & \mathbf{\hat\Gamma}_{RR} = \mathbf{V}\;
  \mathbf{\Gamma}_{RR} \; \mathbf{V}^\dag\;, \nn
& \mathbf{\hat\Gamma}_{LR} = \mathbf{U}\; \mathbf{\Gamma}_{LR} \;
  \mathbf{V}^\dag\;, \qquad & \mathbf{\hat\Gamma}_{LL} = \mathbf{U}\;
  \mathbf{\Gamma}_{LL} \; \mathbf{U}^\dag\;.
\end{align}
It is important to notice that the transformation rules for
propagators and for vertices \eqs{eq:proptrans}{eq:gamtrans},
respectively, are different, which reflects the general difference in
transformation rules for amputated and non-amputated Green's
functions.

Let us examine now the general $n$-point transition amplitude with
fermion lines (external or internal - our formalism applies to the
latter by setting final and initial fermion indices to be equal).
Lets focus on any chosen fermion line in such an amplitude.  The
Feynman rule in the mass eigenstate basis would have the general form,
\bea
{\cal M}_{ji}\sim \left(\mathbf{\Gamma}\;\bm{\Delta}\; \mathbf{\Gamma}
\; \ldots \bm{\Delta}\; \mathbf{\Gamma} \right)_{ji}\;.
\label{eq:mfer}
\eea

Applying flavour rotation to the external states and using
\eqs{eq:proptrans}{eq:gamtrans}), one can get an expression for the
fermion line in flavour basis, $\hat {\cal M}_{JI},$ written as a
sequential product of flavour-basis fermion vertices and matrix
propagators,
\bea 
\bm{\hat{\mathcal{M\,}}} &=& \bm{\bar{\mathcal{U}}}^\dag
\bm{{\mathcal{M\,}}}\;\bm{{\mathcal{U}}}^\dag
\sim \Big(\bm{\bar{\mathcal{U}}}^\dag {\bf
  \Gamma\,}\bm{{\mathcal{U}}}^\dag \Big) \Big(\bm{{\mathcal{U}}} {\bf
  \Delta}\,\bm{\bar{\mathcal{U}}}\Big)
\Big(\bm{\bar{\mathcal{U}}}^\dag {\bf \Gamma\,}\bm{{\mathcal{U}}}^\dag
\Big) \ldots \Big(\bm{{\mathcal{U}}} {\bf
  \Delta}\,\bm{\bar{\mathcal{U}}}\Big)
\Big(\bm{\bar{\mathcal{U}}}^\dag {\bf \Gamma\,}\bm{{\mathcal{U}}}^\dag
\Big) \nonumber\\
&=& {\bf \hat\Gamma\;\hat\Delta\; \hat\Gamma\;\dots \hat\Delta\;
  \hat\Gamma}\;.
\label{Mferflav} 
\eea

This shows that any fermionic amplitude built of vertices and
propagators in the mass basis can be formally transformed into the
amplitude given in terms of respective quantities in the flavour
basis.  We can now observe that, as shown explicitly in \eq{mat-prop},
matrix denominators of the loop integrals in the flavour basis always
depend on Hermitian matrices $\mathbf{M M^\dag}$ or $\mathbf{M^\dag
  M}$, and only such combinations would appear as formal arguments of
loop functions.  As a consequence, one can conclude that FET
formulated for Hermitian matrices can always apply to loop functions
appearing in fermionic amplitudes, as well.

From more practical point of view, our derivation leads to the
conclusion that the fermion mixing matrices $\mathbf{U}$ and
$\mathbf{V}$ can appear in amplitudes only is some specific
combinations, namely
\bea
U_{Bi}\, f(m_i^2) \, U_{Ai}^{\star} &=& f(\mathbf{M^\dag
  M})_{BA}\;,\nonumber\\[1mm]
V_{Bi}\, f(m_i^2) \, V_{Ai}^{\star} &=& f(\mathbf{ M M^\dag
})_{BA}\;,\nonumber\\[1mm]
U_{Bi}\, m_i f(m_i^2) \, V_{Ai}^{\star} &=&{M}_{BC}^\dag\, f(\mathbf{M
  M^\dag})_{CA} = f(\mathbf{M^\dag M})_{BC} \, {M}_{CA}^\dag
\;,\nonumber\\[1mm]
V_{Bi}\, m_i f(m_i^2) \, U_{Ai}^{\star} &=& {M}_{BC}\, f(\mathbf{
  M^\dag M})_{CA} = f(\mathbf{ M M^\dag})_{BC} \, {M}_{CA}\;,
\label{eq:fermionFET}
\eea
which can always be expanded using~\eq{theorem}.  

We should notice that the formal treatment followed in this section
can easily generalize to the case of more complicated flavour models,
where sets of flavour fields belong to distinct flavour families.  In
this case, the propagators and the vertices in the general formulae of
\eqs{eq:mfer}{Mferflav}, will carry both internal (flavour) and
external (family-group) indices.  However, only $\bf{\Gamma}$ or
$\bm{\hat\Gamma}$ can associate different family groups because
propagators, $\mathbf{\Delta}(k)$ or $\mathbf{\hat{\Delta}}(k)$ are
block diagonal in family space.  Therefore, one can accommodate in
this formalism amplitudes with complicated flavour structure like, \eg
rare processes in the MSSM with fermions on the external lines and
sfermions and gauginos circulating in loops.  FET formalism not only
allows one to calculate such diagrams in flavour basis but also in any
other ``hybrid" basis of convenience, \eg fermions-gauginos in mass
and sfermions in flavour basis, or any other combination.

\section{Application of FET: neutron EDM in the MSSM}
\setcounter{equation}{0}
\label{sec:nedm}

To illustrate that higher order mass insertion terms can give
physically meaningful bounds we consider the example of the neutron
Electric Dipole Moment (nEDM) in the Minimal Supersymmetric Standard
Model~\cite{Nilles:1983ge,Haber:1984rc,Martin:1997ns}.

The full nEDM can be expressed as a combination of parton level
contributions - EDMs of quarks $d_q$, their chromoelectric dipole
moments (CDM) $c_q$ and the CDM of gluon $c_g$.  The parton moments
are defined as respective coefficients in the effective Hamiltonian:
\bea 
{\cal H}_q &=& {i d_q\over 2}\bar{q} \sigma_{\mu\nu} \gamma_5 q
F^{\mu\nu} - \frac{i c_q}{2} \bar{q} \sigma_{\mu\nu} \gamma_5 T^a q
G^{\mu\nu a}\;,\nonumber\\
{\cal H}_g &=& - \frac{c_g}{6} f_{abc} G^a_{\mu\rho} G^{b\rho}_{\nu}
G^c_{\lambda\sigma}\epsilon^{\mu\nu\lambda\sigma}.
\label{eq:edm}
\eea
The total neutron EDM depend on its hadronic wave function and can be
written as
\bea
E_n = \eta_{ed} d_d +\eta_{eu} d_u + e(\eta_{cd} c_d + \eta_{cu} c_u)
+\frac{e\eta_g\Lambda_X}{4\pi}c_g\;,
\label{eq:fullneut}
\eea
where $\eta_i$ and $\Lambda_X$ are ${\cal O}(1)$ QCD wave-function
factors~\cite{Fuyuto:2013gla} and the chiral symmetry breaking
scale~\cite{Manohar:1983md}, respectively.  Various models can give
significantly different values for $\eta_i$, differing even by an
overall sign.  Thus, eq.~(\ref{eq:fullneut}) and the bounds it puts on
MSSM parameters should be treated as order of magnitude estimates
only, since potential cancellations in~(\ref{eq:fullneut}) depend on
these poorly known coefficients.

The explicit expressions for $d_q, c_q$ and $c_g$ are given in
ref.~\cite{Pokorski:1999hz}.  In this example, we consider only the
dominant gluino contribution to the up-quark EDM and CDM.  Taken
together, their contribution to the nEDM can be expressed as
\bea
E_n = \frac{1}{M_3}\: \sum_{k=1}^6\: \mathrm{Im}
(Z_{U}^{1k}Z_{U}^{4k\star}) \: F(x_{\tilde{U}_k})\;,
\label{eq:nedmu}
\eea
where $Z_U$ and $m_{\tilde{U}_k}$ are up--squark mixing matrices and
physical masses, $M_3$ is the gluino mass (for conventions and the
detailed definitions see~\Refs{Rosiek:1995kg, Rosiek:1989rs}), and we
define the mass ratios, $x_{\tilde Q} \equiv m_{\tilde{Q}}^2/M_3^2$.
The function $F(x)$ is the sum of loop contributions
\bea
F(x) =\frac{e\alpha_s}{18\pi}\left( 8\eta_{eu} C_{12}(x) -
\frac{3g_s\eta_{cu}}{2}\left(18 C_{11}(x) + C_{12}(x)\right)\right)\;,
\label{eq:nedmf}
\eea
with $C_{11},C_{12}$ being certain PV-functions defined as
\bea
C_{11}(x) & = & {-1 + 3x\over 4(1 - x)^2} + {x^2\over 2(1-x)^3} \log
x\,,
\label{eq:cp11}
\\
C_{12}(x) & = & -{x + 1\over 2(1 - x)^2} - {x\over (1-x)^3} \log x\,.
\label{eq:cp12}
\eea

Flavour violation in the MSSM squark sector is strongly bounded by
numerous experiments and known to be very small, $\lesssim
O(10^{-3})$, for down squark mass matrices if the diagonal elements of
those matrices are around the electroweak scale.  Therefore, we assume
for the purposes of this example that the left down soft SUSY breaking
squark mass matrix is approximately diagonal, but not degenerate, of
the form
\bea
(m_{\tilde D}^2)_{LL}=\left(
\begin{array}{ccc}
m^2_{\tilde D} & 0 & 0 \\
0 & m^2_{\tilde D} + \delta m^2_{D12} & 0 \\
0 & 0 & m^2_{\tilde D} + \delta m^2_{D13} \\
\end{array}
\right)\;.
\eea
In the left up-squark sector the off-diagonal mass terms are then
generated by the $SU(2)$ relation:
\bea
(m_{\tilde U}^2)_{LL} \ =\ K\, (m_{\tilde D}^2)_{LL} \, K^\dagger \;,
\label{eq:kmrel}
\eea
where $K$ denotes the Cabibbo-Kobayashi-Maskawa (CKM) matrix.

Consider now the flavour expansion of eq.~(\ref{eq:nedmu}).  In the
first order it constrains the imaginary part of the trilinear
up-squark mixing,
\bea
E_n^{(1)} \supset -\frac{v_2}{ M_3^3\sqrt{2}} \, \mathrm{Im} (A_U^{11} +
Y_u \mu^* \cot\beta)\, F^{[1]}\left(x_{\tilde{U}_{L1}},
x_{\tilde{U}_{R1}} \right)\;,
\label{eq:nedm1}
\eea
where the RHS is now expressed in terms of parameters in flavour
basis.  In particular, the arguments of the first divided difference,
$F^{[1]}$, are given by diagonal elements of up-squark mass matrix,
\bea 
x_{\tilde{U}_L} \equiv \frac{(m_{\tilde U}^2)_{LL}^{11}}{M_3^2}\,\,\,,\,\,
\, x_{\tilde{U}_R} \equiv  \frac{(m_{\tilde U}^2)_{RR}^{11}}{M_3^2}\,.
\eea 
Following the current experimental bound, $|E_n|<2.9 \times
10^{-26}$~\cite{Baker:2006ts}, and bearing in mind potential
cancellations, \eq{eq:nedm1} sets strong bounds on the imaginary
phases of $\mu$ and $A_U^{11}$, of the order of $10^{-3}$ and
$10^{-5}$, respectively, for SUSY mass scale of the order of 1 TeV.

What is interesting, and to our knowledge has not been discussed thus
far in the literature, is that the experimental bound on nEDM is so
strong that it constrains significantly also the {\em real} parts of
up-squark mass insertions, an effect which is easily visible when
analyzing higher orders in MIA expansion.  To avoid lengthy
expressions, let us assume that in the up-squark sector only the $31$
off-diagonal entries do not vanish in the ``right-handed'' soft mass
matrix $(m_{\tilde U}^2)_{RR}^{31}$ and in the trilinear mixing matrix
$A_u^{31}$ and that they are \emph{purely real}.  In addition, the
$(m_{\tilde U}^2)_{LL}$ is defined by the relation to diagonal down
sector in eq.~(\ref{eq:kmrel}).  Then, a non-vanishing contributions
to nEDM are generated from higher orders in mass insertions via the
mixing with the complex CKM matrix elements.  Using the FET up to 2nd
order one can see that the result depends only on the $A_u^{31}$,
\bea
E_n^{(2)} \supset \frac{v_2\sin 2\theta_{13} \cos\theta_{23}}{
  2\sqrt{2} \: M_3^5} \: (\delta m_{D13}^2 - \delta
m^2_{D12}\sin^2\theta_{12} )\, \sin\delta_{CKM}\, \mathrm{Re} A_u^{31}\, 
F^{[2]}\left(x_{\tilde{U}_{L1}}, x_{\tilde{U}_{R1}},
x_{\tilde{U}_{L3}}\right)\;, \nonumber\\
\label{eq:nedm2}
\eea
where $\theta_{12},\theta_{13},\theta_{23}$, and $\delta_{CKM}$ are
respectively, the angles and the CP-violating phase in the standard
CKM matrix parametrization.  Note again that a CP-violating observable
constraints \emph{real} squark mass parameters through the CKM
CP-violating phase.

It is worth noting that even the 3rd order expansion of FET sets
numerically significant constraints on the real parts of flavour
violating parameters.  In particular, the dependence on $(m_{\tilde
  U}^2)_{RR}^{31}$ parameter, absent at lower orders, is now
introduced through,
\bea
E_n^{(3)} &\supset&\frac{v_2\sin2\theta_{13}
  \cos\theta_{23}}{2\sqrt{2}\: M_3^7} \, (\delta m_{D13}^2 - \delta
m_{D12}^2 \sin^2\theta_{12})\, \sin\delta_{KM} \times \nonumber\\
&\times& \mathrm{Re}(m_{\tilde U}^2)_{RR}^{31} \,\mathrm{Re} (A_U^{33}
+ Y_t\mu^*\cot\beta)\, F^{[3]}\left(x_{\tilde{U_L}_1},
x_{\tilde{U_R}_1}, x_{\tilde{U_L}_3}, x_{\tilde{U_R}_3}\right)\;.
\label{eq:nedm3}
\eea
Comparing separately expressions given in \eqs{eq:nedm2}{eq:nedm3}
with the experimental upper bound on the neutron EDM, one can obtain
order of magnitude estimates on, otherwise weakly constrained, $31$
and $13$ elements of the up-squark trilinear and ``right-handed'' soft
mass terms in relation to mass splitting in the down-squark sector.
Such bounds are important e.g.  for analysis of the maximal allowed
decay rates of the top quark to lighter MSSM Higgs boson, $t\to
u\,h$~\cite{Dedes:2014asa}.
The numerical results for a typical MSSM parameter set, obtained using
the full unexpanded mass eigenstates expressions for nEDM and the
SUSY\_FLAVOR library~\cite{ Rosiek:2010ug,Crivellin:2012jv,
  Rosiek:2012kx, Crivellin:2013zka, Rosiek:2014sia}, are collected in
Table~\ref{tab:nedm}.  They all agree both qualitatively and
quantitatively with \eqs{eq:nedm2}{eq:nedm3} that have been obtained
from the FET of \eq{theorem}.

\begin{table}[t]
\begin{center}
 \begin{tabular}{|c|c|c|c|c|c|}
\hline 
$ \delta m_{D13}\;[TeV]$ & $0.2$ & 0.4 & 0.6 & 0.8 & 1 \\[1mm]
\hline
$ |\mathrm{Re}A_U^{31}/M_3| \hspace{-0.7cm}$ & $2.43 \times 10^{-2} $ &
$1.31 \times 10^{-2}$ &$1.12\times 10^{-2}$ & $8.94\times 10^{-3}$ &
$7.85\times 10^{-3}$ \\[1mm]
%
%
$ |\mathrm{Re}(m_{\tilde{U}}^2)_{RR}^{31} /M_3^2| $ & $2.47\times
10^{-2}$ & $1.23\times 10^{-2}$ & $1.11\times 10^{-2}$ & $8.64\times
10^{-3}$ & $7.41\times 10^{-3}$ \\[1mm]
\hline 
\end{tabular} 
\end{center}
\caption{\small Upper bounds on $|\mathrm{Re}A_U^{31}/M_3|$ and
  $|\mathrm{Re}(m_{\tilde{U}}^2)_{RR}^{31}/M_3^2|$ imposed by current
  experimental constraints from neutron EDM.  Displayed values were
  obtained assuming $(m_{\tilde{U}}^2)_{RR}^{31}=0$ for the 2nd row,
  $A_U^{31}=0$ for the 3rd row and no other sources of the sfermion
  flavour violation.  Other parameters set to $\tan \beta = 4$, common
  SUSY-scale $M_3=1.1$ TeV and a suitable value of $A_U^{33}$ was
  implicitly chosen to satisfy the $125$ GeV Higgs mass constraint.}
\label{tab:nedm}
\end{table}
%
\begin{figure}[t]
    \centering \includegraphics[width=0.8\textwidth]{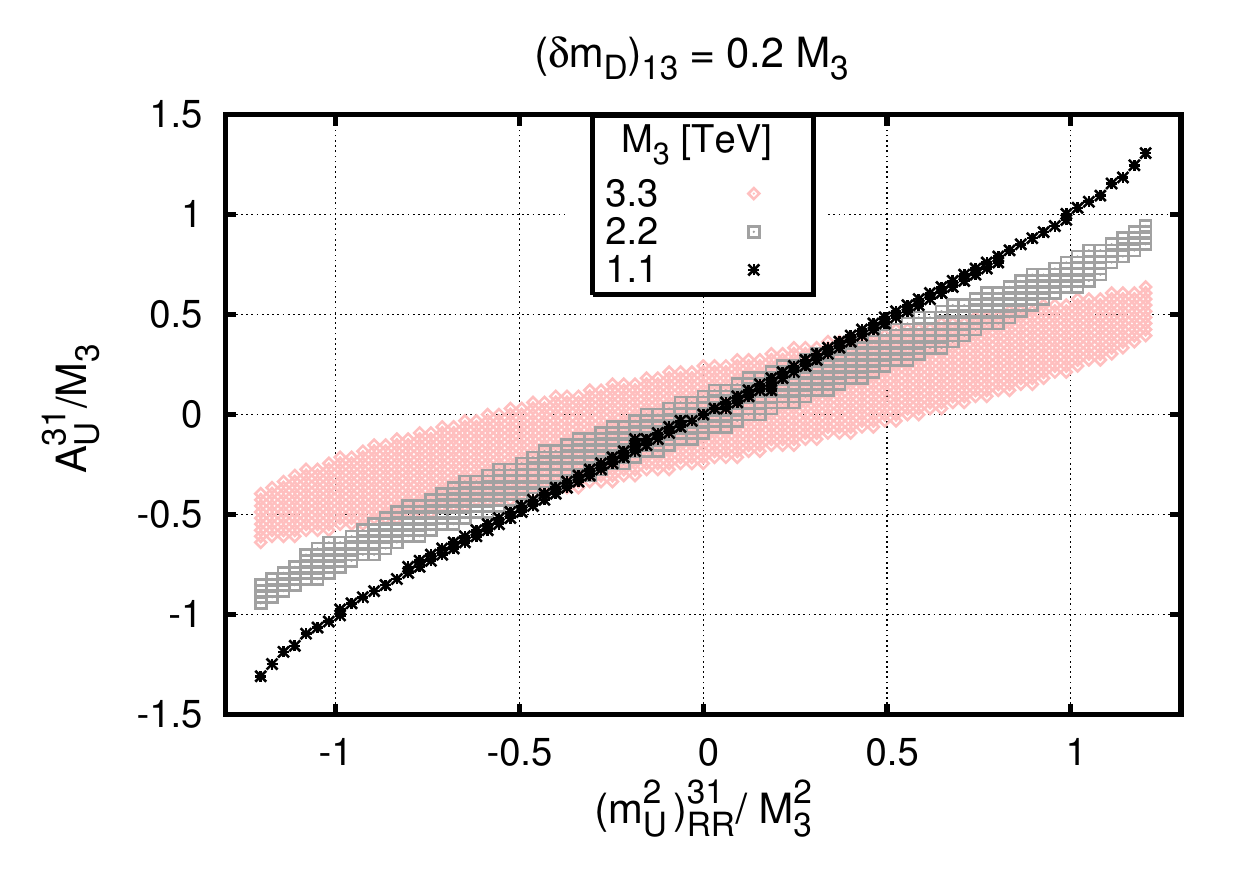}
    \caption{Range of the
      $\mathrm{Re}A_U^{31}/M_3-\mathrm{Re}(m_{\tilde{U}}^2)_{RR}^{31}/M_3^2$
      plane allowed by current experimental constraints from neutron
      EDM.  MSSM parameters defined as in Table~\ref{tab:nedm}.  }
    \label{fig:nedm}
\end{figure}

Alternatively, one can satisfy the nEDM bound by choosing flavour
violating entries large but correlated so that terms in
\eqs{eq:nedm2}{eq:nedm3} cancel each other to large accuracy.  It is
interesting to observe that CKM-related factors in these equations are
identical, so the correlation between $A_U^{31}$ and
$(m_{\tilde{U}}^2)_{RR}^{31}$ is given only by SUSY parameters (of
course once the QCD-related factors in \eq{eq:nedmf} become better
known).  Both terms exactly cancel when
\bea
{\mathrm{Re}[A_u^{31}/M_3]\over \mathrm{Re}[(m_{\tilde
      U}^2)_{RR}^{13}/M_3^2]} = - \frac{A_U^{33} +
  Y_t\mu^*\cot\beta}{M_3} \frac{F^{[3]}\left(x_{\tilde{U_L}_1},
  x_{\tilde{U_R}_1}, x_{\tilde{U_L}_3}, x_{\tilde{U_R}_3}\right)}{
  F^{[2]}\left(x_{\tilde{U_L}_1}, x_{\tilde{U_R}_1},
  x_{\tilde{U_L}_3}\right)}\,.
\label{eq:nedmcanc}
\eea

Eq.~(\ref{eq:nedmcanc}) suggests a linear shape for the allowed
$\mathrm{Re}A_U^{31} - \mathrm{Re}(m_{\tilde{U}}^2)_{RR}^{31}$
parameter space, consistent with the nEDM experimental bound.  This
result is illustrated clearly in Fig.~\ref{fig:nedm}, obtained again
with the \code~ code (\ie without using the MIA expansion), and
assuming values of QCD factors implemented there.  Again this result
follows closely \eqs{eq:nedm2}{eq:nedm3} that have been obtained from
the FET.

Furthermore, we have successfully applied FET to another observable,
namely the rare top decay to light quarks and a Higgs boson, $t\to q
\,h$, in MSSM~\cite{Dedes:2014asa}.  What one practically gains from
using FET in complicated theories, like MSSM for example, is to
algebraically isolate the dominant effects for a given observable
without performing tedious diagrammatic MIA expansion nor extensive
computer scans of a multi-parameters space.  

\section{Summary}
\setcounter{equation}{0}
\label{sec:summary}

In this article we have presented and proved a theorem of matrix
analysis, the Flavour Expansion Theorem (FET), that remarkably
translates any transition amplitude, obtained in terms of mass
eigenstate basis parameters, into its corresponding amplitude in
flavour eigenstate basis, purely algebraically, without the use of
standard diagrammatic methods like the Mass Insertion Approximation
(MIA) method.  Following the formulation of this theorem, any analytic
function of a Hermitian matrix is expanded in terms of its
off-diagonal elements with coefficients being the multi-variable
functions, commonly known as divided differences.  Natural
implementation of such expansion [see \eq{theorem}] comes from the
intimate connection between the divided differences and the
Passarino-Veltman one-loop functions.  Apart from the formal proof, we
have discussed also FET limitations, such as non-analyticity and
convergence issues.  We have also extended the use of the theorem in
case of general transition amplitudes involving fermion mass matrices
which are not necessarily Hermitian.

We have argued many times throughout this article, that the algebraic
derivation of the flavour basis result using FET is substantially
easier, shorter and more compact than the diagrammatic one.  A
pedagogical example is given in Section~\ref{sec:calcs}.  However, we
also illustrate how FET works with a significant example based on
sparticle (gluon-squark) contributions to neutron-EDMs.  In this
example, the use of FET at higher non-trivial orders is capable to set
fairly strong bounds of order $\sim10^{-2}$ on real parts of up-squark
mixing matrix elements from nEDM measurements by assuming that
CP-violation arises only from the CKM-matrix phase.  This FET result
agrees with our exact numerical calculations [see Table~\ref{tab:nedm}
and Fig.~\ref{fig:nedm}~] using \code~ library.  To our knowledge
these bounds are new in the MSSM flavour physics literature and
demonstrate the usefulness of the Flavour Expansion Theorem,
especially, when it applies to complicated models.

\section*{Acknowledgements}

J.~R.  would like to thank University of Ioannina for the hospitality
during his stay there.  His work was supported in part by the Polish
National Science Center under the research grant
DEC-2012/05/B/ST2/02597.  A.~D.  would like to thank A. Romanino for
useful communication.  M.~P.  would like to thank Klaus Bering for
valuable discussions.  This research has been co-financed by the
European Union (European Social Fund - ESF) and Greek national funds
through the Operational Program ``Education and Lifelong Learning" of
the National Strategic Reference Framework (NSRF) - Research Funding
Programs: THALIS and ARISTEIA - Investing in the society of knowledge
through the European Social Fund.

\newpage

\appendix
\renewcommand{\thesection}{Appendix~\Alph{section}}
\renewcommand{\thesubsection}{\Alph{section}.\arabic{subsection}}
\renewcommand{\theequation}{\Alph{section}.\arabic{equation}}

\section{Proof of the Flavour Expansion Theorem}
\setcounter{equation}{0}
\label{app:fet}
The FET theorem formulated in Section~\ref{sec:fetdef} can be proved
using mathematical induction and the notion of the ``fully symmetrized
polynomials''.
 
\begin{subsection}{Fully symmetrized polynomials.  }

The \emph{``fully symmetrized'' homogeneous polynomials of degree
  $N$}~\cite{fulton1991representation}, can be understood through the
following equivalent definitions:
\begin{def1} $Q^M_N(x_1,\ldots,x_M)$ is the direct sum of all
    distinct N-degree monomials constructed out of the given set of
    $M$ variables $x_i$.
\end{def1}   

\begin{def1}
Alternatively, $Q^M_N(x_1,\ldots,x_M)$ is defined as
\begin{align}
Q^M_N(x_1,\ldots,x_M) \equiv\quad\osum{N} x_1^{a_1}\ldots x_M^{a_M}
\equiv \sum_{a_1,\ldots,a_M = 0}^{N}\Big(x_1^{a_1}\dots
x_M^{a_M}\;\Big)\Big|_{{{a_1+\ldots+a_M}=N}} \;\;\;.
\label{QNM}
\end{align}
\end{def1}
Directly from definitions above, one can express the fully symmetrized
polynomial, $Q^M_N$, for any value of $M,N$.  Due to the symmetric
nature of \eq{QNM} there exist many equivalent representations of this
sum.  For non trivial $M,N$, one of these will have the explicit form
\small\begin{align}
Q^M_N(&x_1,\ldots,x_M) = \nonumber \\[1mm]
& (x_1)^N\ +\ (x_1)^{N-1}(x_2+x_3+\ldots)\ +
(x_1)^{N-2}\Big(x_2^2+\ldots+(x_2x_3)+\ldots\Big)+\ldots\ +\ x_1(\dots)\nonumber\\
&\hspace{3cm}+ \quad(x_2)^N\ +\ (x_2)^{N-1}(x_3+\ldots)\ +\ \ldots
\nonumber\\
&\hspace{5cm}\ddots \nonumber \\
&\hspace{6cm}+\ (x_M)^N\;.  \label{QNMx}
\end{align}\normalsize
The identities, $Q_N^0=0$, $Q_0^M=1$, $Q_N^1(x_1)=x_1^N$, also hold
trivially by definition.

Due to \eq{QNM} the factorization property,
\begin{align}
Q^M_N(x_1,\dots ,x_M) \, =\ \sum_{K=0}^N Q^{L}_K(x_1,\dots,x_L)
\:Q^{M-L}_{N-K}(x_{L+1},\ldots,x_M)\;,\qquad (no~sum~over~L)\,
, \label{QNMfac1}
\end{align}
holds for any integer $L$, satisfying $1\leq L \leq M-1$, and for any
choice of, $\{L\}$ and $\{M-L\}$, respective subsets of $M$ variables,
$x_i$.\\[1mm]
\begin{lem}
Fully symmetrized polynomials $Q^{M+1}_N$ are $M$-order divided
differences of the generating functional, $Q_{N+M}^1$, thus satisfying
$Q_N^{M+1}=Q^{1\;[M]}_{N+M}$.  Equivalently, the expression,
\begin{align}
Q_N^{M+1}=Q^{M\;[1]}_{N+1}\;,
\label{lemma} 
\end{align}
holds for any $M\geq 1$.
\end{lem}

\begin{proof}
First, we show the validity of \eq{lemma}, for $M=1$, namely
\begin{equation}
Q^2_{N}=Q^{1\;[1]}_{N+1}\;.
\label{lemmaM1}
\end{equation}
Applying \eq{QNMfac1}, we have
\begin{align}
&(x_1-x_2)Q^2_{N}(x_1,x_2)= \sum_{K=0}^N \Big(\; Q^{1}_{N-K+1}(x_1)
  Q^{1}_{K}(x_2)- Q^{1}_{K}(x_1) Q^{1}_{N-K+1}(x_2)\;\Big)\nonumber\\
&=Q^{1}_{N+1}(x_1)-Q^{1}_{N+1}(x_2)+\sum_{K=1}^N Q^{1}_{N-K+1}(x_1)
  Q^{1}_{K}(x_2)-\sum_{K=1}^N Q^{1}_{K}(x_1)
  Q^{1}_{N-K+1}(x_2)\nonumber\\
&=Q^{1}_{N+1}(x_1)-Q^{1}_{N+1}(x_2)\label{lemmaM11}\;,
\end{align}
which is equivalent to \eq{lemmaM1}.   Now it is straightforward to
verify \eq{lemma} for $M>1$, as well.   Denoting
$y\equiv\{x_3,\dots,x_{M+1}\}$, we have
\begin{align}
&(x_1-x_2)Q^{M+1}_{N}(x_1,x_2,y)\overset{\eqref{QNMfac1}}{=}
  (x_1-x_2)\sum_{K=0}^N Q^{2}_{K}(x_1,x_2)
  Q^{M-1}_{N-K}(y)\;\nonumber\\
&\hspace{-0.2cm}\overset{\eqref{lemmaM11}}{=}\sum_{K=0}^N
  \Big(Q^{1}_{K+1}(x_1)-Q^{1}_{K+1}(x_2)\Big)Q^{M-1}_{N-K}(y)\nonumber\\
&=\sum_{K=1}^{N+1} \Big(Q^{1}_{K}(x_1)-Q^{1}_{K}(x_2)\Big)
  Q^{M-1}_{N-K+1}(y)\ +\ \Big(Q^{1}_{0}(x_1)-Q^{1}_{0}(x_2)\Big)
  Q^{M-1}_{N+1}(y)\nonumber\\
&=Q^{M}_{N+1}(x_1,y)-Q^{M}_{N+1}(x_2,y) \;,
\end{align}
and therefore finishing the proof of the lemma.
\end{proof}

\end{subsection}

\begin{subsection}{Flavour Expansion Theorem: the proof}

\begin{proof}
Due to Def.~\ref{def1}, the Hermitian matrix function $f(\mathbf{A})$,
can be expressed as a power series,
\begin{equation}
f(\mathbf{A})\ = \ \sum_{m=0}^\infty
c_m \mathbf{A}^m \;.\hspace{2cm}
\label{fAexp}
\end{equation}
One can apply the matrix decomposition $\mathbf{A} = \mathbf{A_0}
+ \mathbf{\hat{A}}$ to the above series (convergent by assumption) and
rearrange terms collecting together the same powers of
$\mathbf{\hat{A}}$.  Assuming that the resulting summation remains
convergent, we have
\begin{align}
f(\mathbf{A})& =\ c_0 \,\mathbf{I} \ +\ c_1 \mathbf{A_0}\;\ +\ \;c_2
\:\mathbf{A}_\mathbf{0}^2\ \; + \ \;\;\; c_3\:
\mathbf{A}_\mathbf{0}^3\;\;\,\ + \ \:\dots &:\: \bf{F_{0}} \nonumber
\\[1mm]
&\hspace{1.3cm}+\ \;c_1 \mathbf{\hat{A}}\ +\ c_2
\langle\mathbf{\hat{A}}\mathbf{A}_\mathbf{0}\rangle\ + \ c_3
\langle\mathbf{\hat{A}} \mathbf{A}_\mathbf{0}^2 \rangle\ + \ \ c_4
\langle\mathbf{\hat{A}} \mathbf{A}_\mathbf{0}^3 \rangle\ + \ \dots &
:\: \bf{F_{1}}\nonumber \\[1mm]
&\hspace{3.cm}+\ \;\,c_2\; \mathbf{\hat{A}}^2\;\ +\ \; c_3
\langle\mathbf{\hat{A}}^2\mathbf{A}_\mathbf{0}\rangle\, + \ c_4
\langle\mathbf{\hat{A}}^2 \mathbf{A}_\mathbf{0}^2 \rangle\ + \ \dots &
:\: \bf{F_{2}}\nonumber \\[1mm]
&\ddots \hspace{4cm}& \vdots \hspace{0.65cm} \nonumber \\[1mm]
& + \ c_M \mathbf{\hat{A}}^M \ + \ c_{M+1}
\langle\mathbf{\hat{A}}^{M}\mathbf{A}_\mathbf{0}\rangle \ +
\ \dots\ +\ \;\; c_n \langle\mathbf{\hat{A}}^{M}
\mathbf{A}_\mathbf{0}^{n-M}\rangle\;\ + \ \dots \: & :\:
\bf{F_{M}}\nonumber \\[2mm]
 & \hspace{3.0cm} + \ c_{M+1}\;
\mathbf{\hat{A}}^{M+1}\ +\ \dots\ +\ c_n \langle\mathbf{\hat{A}}^{M+1}
\mathbf{A}_\mathbf{0}^{n-M-1}\rangle\ +\ \dots &:\:
\bf{F_{M+1}}\nonumber \\
&\hspace{4.cm}\ldots & \ldots\hspace{0.25cm}
\label{fAexpM}
\end{align}
where we have defined
\bea
\langle\mathbf{\hat{A}}^m\mathbf{A}_\mathbf{0}^n\rangle \equiv \sum_{P
  -\textrm{distinct}}\mathbf{\hat{A}}^m \mathbf{A}_\mathbf{0}^n\;,
\eea
for all distinct permutations of the set $\{{\mathbf{\hat{A}}_{\,} \:
  \dots\,} ,\;{\mathbf{A}_\mathbf{0}\dots }\}$ of $(m+n)$ objects.

The matrix element $\bra{I}f(A)\ket{J}$ will be given by the sum
\begin{eqnarray}
\big(f(A)\big)^{IJ} = \sum_{N=0}^{\infty} F_{N}^{IJ}
\ =\ F_{0}^{IJ}\ +\ F_{1}^{IJ}\ + \ F_{2}^{IJ}\ +\ \dots \;,
\end{eqnarray}
where, by direct calculation, the above terms are given by (summation
over repeated internal indices $K_i$, is considered - also if they
appear more than twice),
\begin{subequations}
\begin{align}
F_{0}^{IJ} & =  \delta^{IJ} f(A_0^I)\;, \\[2mm]
F_{1}^{IJ} & = \hat{A}_{IJ} \Big( c_1 \ +\ c_2 \left[A_0^I +
  A_0^J\right] \ +\ c_3 \left[(A_0^I)^2 + (A_0^J)^2 + A_0^I
  A_0^J\right] + \ldots \Big) \;, \\[2mm]
F_{2}^{IJ} & = \hat{A}_{IK_1} \hat{A}_{K_1J} \Big( c_2 \ +\ c_3
\left[A_0^I + A_0^J + A_0^{K_1}\right] \nonumber\\[2mm]
& \ +\ c_4 \left[(A_0^I)^2 + (A_0^J)^2 + (A_0^{K_1})^2 + A_0^I A_0^J +
  A_0^I A_0^{K_1} + A_0^J A_0^{K_1}\right]\ + \ \ldots \Big)
\nonumber\\
&= \hat{A}_{IK_1}\hat{A}_{K_1J}\: \sum_{N=0}^\infty c_{2+N}
Q_N^{3}(A_0^I,A_0^J,A_0^{K_1}) \;, \\
\ldots\hspace{2mm} & \hspace{7cm}\nonumber\\
\,F_M^{IJ} & = \hat{A}_{IK_1}\hat{A}_{K_1 K_2}\ldots
\hat{A}_{K_{M-1}J}\: \sum_{N=0}^\infty c_{M+N}
Q_N^{M+1}(A_0^I,A_0^J,A_0^{K_1},\ldots,A_0^{K_{M-1}}) \;,\\
\ldots \;.&&\nonumber
\end{align}
\label{eq:qexp}
\end{subequations}
To prove the theorem, we need to show
\begin{equation}
\sum_{N=0}^\infty c_{M+N} Q_N^{M+1}(A_0^I,A_0^J,A_0^{K_1}, \ldots,
A_0^{K_{M-1}}) = f^{[M]}(A_0^I,A_0^J,A_0^{K_1},\ldots,A_0^{K_{M-1}})\;,
\end{equation}
for all $M\geq 0$.  This can be realized using mathematical induction.
For $M=0$, we obtain trivially
\begin{align}
&\sum_{N=0}^\infty c_{N} Q_N^{1}(A_0^I)=\sum_{N=0}^\infty c_{N}
  (A_0^I)^N=f(A_0^I)\equiv f^{[0]}(A_0^I)\;.
\end{align}
Now, let
\begin{align}
&\sum_{N=0}^\infty c_{M+N}
Q_N^{M+1}=f^{[M]}\;,
\end{align}
holds for some $M>0$ and for any set of $M+1$ arguments.  Then, we
need to show that
\begin{align}
&\sum_{N=0}^\infty c_{M+N+1}
  Q_N^{M+2}(A_0^I,A_0^J,A_0^{K_1},\ldots,A_0^{K_{M}}) =
  f^{[M+1]}(A_0^I,A_0^J,A_0^{K_1},\ldots,A_0^{K_{M}})\;,
\end{align}
which, by Def.~\ref{def2} of divided differences in \eq{dddef}, is
equivalent to showing
\begin{align}
(A_0^I-A_0^{J})\sum_{N=0}^\infty c_{M+N+1}\, &
  Q_N^{M+2}(A_0^I,A_0^J,A_0^{K_1}, \ldots,
  A_0^{K_{M}})\hspace{5cm}\nonumber\\
&\overset{\eqref{lemma}}{=}\sum_{N=0}^\infty c_{M+N+1}
\Big(Q_{N+1}^{M+1}(A_0^I,A_0^{K_1},\ldots) -
Q_{N+1}^{M+1}(A_0^J,A_0^{K_1},\ldots)\Big)\nonumber\\
&=\sum_{N=1}^\infty c_{M+N} \Big(Q_{N}^{M+1}(A_0^I,A_0^{K_1},\ldots) -
Q_{N}^{M+1}(A_0^J,A_0^{K_1},\ldots)\Big)\nonumber \\
&+\ c_M \Big(Q_{0}^{M+1}(A_0^I,A_0^{K_1},\ldots) -
Q_{0}^{M+1}(A_0^J,A_0^{K_1},\ldots)\Big)\hspace{1.15cm}\nonumber\\
&=\sum_{N=0}^\infty c_{M+N} \Big(Q_{N}^{M+1}(A_0^I,A_0^{K_1},\ldots) -
Q_{N}^{M+1}(A_0^J,A_0^{K_1},\ldots)\Big)\nonumber\\
&=f^{[M]}(A_0^I,A_0^{K_1},\ldots) -
f^{[M]}(A_0^J,A_0^{K_1},\ldots)\;, \hspace{2.75cm}
\end{align} 
and hence the theorem is proved.
\end{proof}

\end{subsection}

\section{Convergence criterion for FET expansion of the one-loop functions}
\setcounter{equation}{0}
\label{app:pvconv}

It is well known that, any one-loop amplitude can be expressed as a
linear combination of ``master'' PV-integrals with trivial \ie equal
to 1, integrand numerator.  Thus, it is sufficient to find a
convergence criterion for the FET expansion only for master integrals.
Below we formulate such a criterion for the most often considered case
of loop functions with vanishing external momenta.  The same criterion
can be applied to coefficients of the expansion of one-loop integrals
in terms of external momenta (assuming that they are far from
thresholds and momentum expansion can be performed) - such
coefficients can be also reduced to combinations of master integrals
with vanishing momenta.

For vanishing external momenta master integrals can be expressed as
\bea
PV^{(n)}_0(m_1^2,\dots,m_{n}^2) = -i(4\pi)^2\int
\frac{d^4p}{(2\pi)^4}\frac{1}{\prod_{j=1}^{n} (p^2-m_{j}^2)} = (-1)^n
\int_0^\infty \frac{u du}{\prod_{j=1}^{n} (u+m_{j}^2)}\;,
\label{eq:PVn0}
\eea
where we assume $n\geq 3$ to avoid divergent integrals - considering
the estimates for finite ones is sufficient to establish the
convergence criterion for FET expansion as it depends only on
behaviour of higher order terms.

Eq.~(\ref{eq:PVn0}) leads immediately to inequality
\bea
\left|PV^{(n+1)}_0(m_1^2,\dots,m_{n}^2,m_{n+1}^2)\right| \leq
\frac{1}{m_{n+1}^2}\left|PV^{(n)}_0(m_1^2,\dots,m_{n}^2)\right| \;.
\label{eq:PVn1}
\eea
Applying this inequality, iteratively for higher order terms, to
majorize the RHS of \eq{theoremPV} (in what follows we do not write
explicitly any PV-function arguments apart from the ones used in the
expansion), we get
\bea
\left|\left[PV^{(n)}_0(\mathbf{A})\right]_{IJ}\right| & \leq & \left|
\delta_{IJ}PV^{(n)}_0(A^I_0)\right| + \left|
PV^{(n+1)}_0(A^I_0,A_0^J)\hat{A}_{IJ}\right| \nn
&&+ \left| PV^{(n+2)}_0(A^I_0,A_0^J,A_0^K)\hat{A}_{IK}
\hat{A}_{KJ}\right| + \ldots\nonumber\\[2mm]
& \leq & \left| \delta_{IJ}PV^{(n)}_0(A^I_0)\right| + \left|
PV^{(n+1)}_0(A^I_0,A_0^J)\right|\left|\hat{A}_{IJ}\right| \nn
&&+ \left|PV^{(n+2)}_0(A^I_0,A_0^J,A_0^K)\right|
\left|\hat{A}_{IK}\right|\left| \hat{A}_{KJ}\right| + \ldots
\nonumber\\[2mm]
&\leq& \left|PV^{(n)}_0(A^I_0)\right|\left( \delta_{IJ} +
\frac{\left|\hat{A}_{IJ}\right| }{A_0^J}+
\frac{\left|\hat{A}_{IK}\right|}{A_0^K}
\frac{\left|\hat{A}_{KJ}\right|}{A_0^J} + \ldots \right) \\[2mm]
&=& \left|PV^{(n)}_0(A^I_0)\right|\left( \delta_{IJ} +
\sqrt{\frac{A_0^I}{A_0^J}} \left(
\frac{\left|\hat{A}_{IJ}\right|}{\sqrt{A_0^IA_0^J}} +
\frac{\left|\hat{A}_{IK}\right|}{\sqrt{A_0^IA_0^K}}
\frac{\left|\hat{A}_{KJ}\right|}{\sqrt{A_0^KA_0^J}} + \ldots
\right)\right)\;,\nonumber
\label{eq:PVestim}
\eea
where we assume that all indices apart from $I,J$ are implicitly
summed in the range $1\ldots N$.  Let us now define the symmetric
matrix $\mathbf{Q}$ with elements being the absolute values of
dimensionless quantities commonly referred in the literature as ``mass
insertions'' (diagonal elements of $\mathbf{Q}$ vanish by definition
of $\mathbf{\hat A}$ matrix)
\bea
Q_{IJ} = \frac{|\hat{A}_{IJ}|}{\sqrt{A_0^I A_0^J}}\;.
\label{eq:deldef}
\eea
Then \eq{eq:PVestim} can be expressed as
\bea
\left|\left[PV^{(n)}_0(\mathbf{A})\right]_{IJ}\right| & \leq &
\left|PV^{(n)}_0(A^I_0)\right|\left( \delta_{IJ} +
\sqrt{\frac{A_0^I}{A_0^J}} \left(\mathbf{Q} + \mathbf{Q}^2 + \ldots
\right)_{IJ} \right)\;.
\label{eq:PVdel}
\eea
The expression in the inner parenthesis of the RHS of \eq{eq:PVdel} is
a geometric series.  According to the definition of a function of
Hermitian matrix given in Section~\ref{sec:math}, this series is
convergent if it converges also for any of $\mathbf{Q}$ eigenvalues,
hence their absolute values must be all smaller than 1.  This can be
expressed formally, as
\bea
 \sup_{|| \mathbf{e} ||=1}
 |\mathbf{e}^\intercal\,\mathbf{Q}\,\mathbf{e}|\ = \sup_{|| \mathbf{e}
   ||=1} |\mathbf{e}^\intercal\,\mathbf{D}_Q\,\mathbf{e}|\ <\ 1\;,
\eea
where $\mathbf{e}$ denotes any real unit vector, and $\mathbf{D}_Q$ is
the diagonal matrix of eigenvalues.  Obviously, this is a sufficient
but not necessary condition for the convergence of the expansion.

Finally we should note that vanishing diagonal elements can not pose a
threat for the convergence of the FET expansion in physical theories.
This is because all Hermitian (squared) mass matrices are
semi-positive definite matrices, and for such matrices if $A_0^I = 0$,
then necessarily also $\hat{A}_{IK}=\hat{A}_{KI}=0$ for all $K$.  Thus
all potentially divergent terms vanish.

\newpage

\bibliography{fet}{}
\bibliographystyle{JHEP}

\end{document}